\documentclass[a4paper, 12pt]{article}
\usepackage{amsmath,amssymb,amsthm}
\usepackage{authblk}
\usepackage[T1]{fontenc}
\usepackage{bbm}
\usepackage{color}

\newcommand{\be}{\begin{equation}}
\newcommand{\ee}{\end{equation}}

\newcommand{\lel}{\pl = \pl}

\newcommand{\pl}{\hspace{.1cm}}

\DeclareMathOperator{\flip}{flip}

\def\Id{{\mathbb I}}
\def\H{{\cal H}}
\def\Cx{{\mathbb C}}     
\def\norm #1{\Vert #1\Vert}

\mathchardef\ree="023C \mathchardef\imm="023D  

\newcommand{\la}{\lambda}

\def\BB{{\mathcal B}}
\def\HH{{\mathcal H}}

\def\AA{{\mathcal A}}

\def\CC{{\mathcal C}}
\def\FF{{\mathbb F}}

\def\omin{{\otimes_{min}}}
\def\omax{{\otimes_{max}}}
\def\ocmax{{\otimes_{cmax}}}
\def\star{{*}}

\newtheorem{thm}{Theorem}
\newtheorem{lem}[thm]{Lemma}
\newtheorem{prop}[thm]{Proposition}
\newtheorem{cor}[thm]{Corollary}
\newtheorem{defn}[thm]{Definition}
\newtheorem{remark}[thm]{Remark}

\title{Connes' embedding problem and Tsirelson's problem}

\author[1]{M. Junge}
\author[2]{M. Navascues}
\author[1]{C. Palazuelos}
\author[2]{\authorcr D. Perez-Garcia}
\author[3]{V. B. Scholz}
\author[3]{R. F. Werner}
\affil[1]{\small Department of Mathematics, University of Illinois at
Urbana-Champaign, Illinois 61801-2975, USA}
\affil[2]{\small Departamento de An\'{a}lisis Matem\'{a}tico and IMI,
Universidad Complutense de Madrid, 28040, Madrid, Spain}
\affil[3]{\small Institut f\"{u}r Theoretische Physik, Leibniz Universit\"{a}t Hannover, Appelstr. 2, 30167 Hannover, Germany}

\begin{document}
\maketitle
\abstract{We show that Tsirelson's problem concerning the set of quantum correlations and Connes' embedding problem on finite approximations in von Neumann algebras (known to be equivalent to Kirchberg's QWEP conjecture) are essentially equivalent. Specifically, Tsirelson's problem asks whether the set of bipartite quantum correlations generated between tensor product separated systems is the same as the set of correlations between commuting C*-algebras. Connes' embedding problem asks whether any separable II$_1$ factor is a subfactor of the ultrapower of the hyperfinite II$_1$ factor. We show that an affirmative answer to Connes' question implies a positive answer to Tsirelson's. Conversely, a positve answer to a matrix valued version of Tsirelson's problem implies a positive one to Connes' problem.}

\newpage

\section{Introduction}

In non-relativistic Quantum Mechanics, measurements conducted on a quantum system by two distant observers are usually modeled by operators acting on a tensor product of two Hilbert spaces. Each factor corresponds to one observer, and the action of a measurement operator is assumed to be non-trivial only on its associated party's space. In contrast, in Algebraic Quantum Field Theory (AQFT \cite{haag}), local observables are represented by operators acting on a joint Hilbert space, and the independence condition reduces to demanding that operators corresponding to different parties commute.

The problem of Tsirelson is to decide whether these two mathematical models give rise to the same set of probability distributions. In other words, is it possible to represent all bipartite probability distribution originating from commuting sets of observables by using observables of tensor product form? The problem originated in a premature claim (positive answer to the question) in \cite{tsirelson}, for which the authors of \cite{Navascues:2007lq} demanded a proof. Tsirelson then posted it on the open problems site at Braunschweig \cite{B.S.Tsirelson:2006fk}.

A negative answer to this question would allow, in principle, to demonstrate experimentally that finite dimensional quantum models do not suffice to describe all bipartite correlations. On the other hand, a positive answer amounts to saying that the powerful numerical algorithms to limit the commutative set of correlations \cite{Navascues:2007lq}, \cite{Navascues:2008} produce not just upper bounds to quantum correlations, but best upper bounds.

Though its original motivation stems from physical considerations, Tsirelson's problem is closely related to finite approximability in operator algebras \cite{ScholzWerner}. The most prominent question in this field is the Connes' embedding problem for von-Neumann algebras (see \cite{MR2072092,MR2039135,Pisier_book} for nice reviews about it). It asks whether any separable II$_1$-factor is a subfactor of the ultrapower $R^\omega$ of the hyperfinite II$_1$-factor $R$. This problem, casually raised by Connes, has many equivalent formulations. One of them, whether the predual of any separable von Neumann algebra is finitely representable in the trace class $S_1$, is related to the possible extension of local Banach space theory to its non-commutative relative: operator spaces. The one we will use here is known as Kirchberg's QWEP conjecture and it asks whether all C*-algebras are quotients of C*-algebras with the weak expectation property. This can be read also as the existence of a unique C*-algebra norm in the tensor product $C^*(\mathbb{F}_n)\otimes C^*(\mathbb{F}_n)$. As shown for instance in the review \cite{MR2072092}, both a positive and negative solution to Connes' problem would have deep implications:

A positive solution would lead to new results concerning invariant subspaces. It would shed new light on the conditions under which the semigroup ${\rm Ext}(A)$ of a C*-algebra $A$ is indeed a group. It would also show that all countable discrete groups are hyperlinear, refuting the famous conjecture of Gromov that ``any statement that holds for all countable discrete groups is either trivial or false''. A negative solution would also have nice applications. For instance in free probability, where it would imply that the two possible definitions of free entropy do not coincide.

In this paper, we show that a positive answer to Connes' embedding problem would also imply an affirmative answer to Tsirelson's problem. Furthermore, we will also show that the converse also holds in some sense. More precisely, we will show the following theorem.

\begin{thm}\label{FirstTheorem}
   If Kirchberg's QWEP conjecture holds, then every probability distribution between two parties which can be represented using commuting sets of observables could also be represented by observables of tensor product form. Conversely, if Tsirelson's problem has a positive solution, also in the case of matrix valued coefficients (made clear below), then the QWEP conjecture, and hence Connes' embedding problem, is also true.
\end{thm}

The structure of this paper is as follows: in Sections \ref{Tsir_prob} and \ref{secqwep} we will state precisely both mathematical problems. In Section \ref{SecNSG} we will introduce essential operator space notation that we will use in the Section \ref{Sec_equiv} to prove Theorem \ref{FirstTheorem}. Along the way, we will obtain an intermediate result that is important in its own right: namely, that Tsirelson's problem is independent of whether we restrict measurements to be projective or general Positive Operator Valued Measures (POVMs).

\section{Tsirelson's problem}
\label{Tsir_prob}

As taught in standard Quantum Mechanics textbooks \cite{tanoudji_book}, a measurement $z$ with $K$ possible outcomes is described by a collection of projector operators $\{E^z_c:c=1,...,K\}$ acting over a given Hilbert space $\H$, and such that

\begin{enumerate}

\item
$E^z_cE^z_{c'}=E^z_c\delta_{cc'}$ (Orthogonality),
\item
$\sum_{c=1}^nE^z_c=\Id$ (Completeness).

\end{enumerate}

\noindent A physical state $\omega$ is mathematically represented by a positive linear functional $\omega:B(\H)\to \Cx$ satisfying the normalization condition $\omega(\Id)=1$. The probability of obtaining an outcome $c$ when measuring property $z$ on a quantum system in the state $\omega$ is given by the expression

\be
p(c|z)=\omega(E^z_c).
\ee

\noindent Note that these probabilities are normalized, since

\be
\sum_{c=1}^Kp(c|z)=\sum_{c=1}^K\omega(E^z_c)=\omega(\sum_{c=1}^KE^z_c)=\omega(\Id)=1.
\ee

The above is the standard description of the measurement process, and it is common to all present formulations of Quantum Mechanics as long as just one observer is involved in the quantum experiment. In the bipartite scenario, though, there are different options for identifying the ``parts'' of the system, which are supposedly under control of different characters, say Alice and Bob.

In non-relativistic Quantum Mechanics, measurement operators associated to different parties are assumed to act over different Hilbert spaces. More specifically, if we call the observers Alice and Bob, then Alice's (Bob's) measurement of property $x$ ($y$) is linked to the set of operators\footnote{For simplicity we will assume that Alice and Bob's measurements all have $K$ possible outcomes.} $\{E^x_a:a=1,...,K\}\subset B(\H_A)$ ($\{E^y_b:b=1,...,K\}\subset B(\H_B)$), which are required to satisfy orthogonality and completeness relations on the space $\H_A$ $(\H_B)$.

States in this scenario are normalized positive functionals of $B(\H_A\otimes\H_B)$, and the probability that Alice and Bob respectively observe the outcomes $a,b$ when they perform measurements $x,y$ is equal to

\be
p(a,b|x,y)=\omega(E^x_a\otimes E^y_b).
\label{tensor}
\ee

Fix $K$ and assume that Alice and Bob are each able to interact with their system in $N$ possible different ways, and so the index $x$ and $y$ run from 1 to $N$. Following Tsirelson's notation \cite{tsirelson}, any set of probabilities $\{p(a,b|x,y):a,b=1,...,K;x,y=1,...,N\}$ will be called \emph{behavior}. We will define $Q$ as the set of all behaviors for which there exist operators $\{E^x_a,E^y_b\}$ satisfying conditions 1,2 and a state $\omega$ such that equation (\ref{tensor}) holds. Sometimes we will refer to $Q$ as the \emph{tensor set of quantum correlations}.

In AQFT, i.e., in relativistic quantum field theory, bipartite correlations are described in a prima facie more general way:  This time, measurement operators associated to Alice and Bob act over the same Hilbert space $\H$, and measurement operators acting on different sites commute with each other \cite{AQFT}. In short, Alice's and Bob's operators must satisfy:

\begin{enumerate}
\item
$E^z_cE^z_{c'}=E^z_c\delta_{cc'}$ (Orthogonality),
\item
$\sum_{c=1}^KE^z_c=\Id$ (Completeness),
\item
$[E^x_a,E^y_b]=0$, for all $a,b=1,...,K;x,y=1,...,N$ (Microcausality).
\end{enumerate}

\noindent For any state $\omega$, the probability that Alice and Bob observe the results $a,b$ when they perform measurements $x,y$ is given by

\be
p(a,b|x,y)=\omega(E^x_a\cdot E^y_b).
\label{commute}
\ee

In analogy with $Q$, we will denote by $Q'$ the set of all behaviors of the form (\ref{commute}). Note that, since $[E^x_a\otimes \Id_B,\Id_A\otimes E^y_b]=0$, $Q$ is contained in $Q'$. We will call $Q'$ the \emph{commutative set of quantum correlations}.

Tsirelson's problem consists in determining whether the inclusion $Q\subset Q'$ is strict, i.e., whether the sets of correlations predicted by the two ways of formalizing ``subsystems'' might differ.

To establish a connection with previous literature on the subject, note that the sets $Q, Q'$ do not grow when we just demand positivity rather than orthogonality for Alice and Bob's measurement operators. Indeed (see Remark \ref{vonNeumann}), if we relax the orthogonality condition to

\begin{enumerate}
\item
$E^z_c\geq 0$ (Positivity)
\end{enumerate}
in the definitions of $Q,Q'$, then the resulting sets of correlations are again $Q$ and $Q'$. Actually, these last definitions for $Q$ and $Q'$ appear in the original formulation of Tsirelson's problem \cite{ScholzWerner}.

We note that the possible difference between $Q$ and $Q'$ is not so much a consequence of relativistic space-time physics, but of the necessity to handle infinitely many degrees of freedom in AQFT. Indeed, one usually imposes axioms to tame the number of local degrees of freedom, e.g., by nuclearity constraints \cite{nukes}. This implies that local algebras of strictly separated space-time regions are contained in Hilbert space tensor factors (``split property'' \cite{DoLo}). Since the separation can be arbitrarily small, this forces the local algebras to be hyperfinite \cite{BuDaFre}, and from the discussion below it is clear that this excludes all correlations which are stronger than tensor products allow. This conclusion holds even in the scenario, when especially strong, state-independent violations of Bell inequalities have been demonstrated \cite{SumWer}, namely when space-time regions touch.

One may also strengthen the notion of ``local subsystems'' by demanding that in a local labs scenario Alice and Bob should not only be able to choose measurements independently of each other, but should also be able to prepare local states as needed  (see also\cite{FloSum}). Again, this implies the split property \cite{LocPrep}, and excludes $Q'$.

Of course, a negative answer to Tsirelson's problem would mean that, if an appropriate correlation expression could be constructed and implemented in the laboratory, the split property could be refuted experimentally.

\section{Connes' embedding problem}\label{secqwep}

We will describe the problem in terms of the universal C*-algebra of the free group with $n$ generators $C^*(\mathbb{F}_n)$. That this is indeed a reformulation of Connes' embedding problem is a deep result of Kirchberg \cite{Kirchberg,MR2072092}. The free group $\FF_n$ with $n$ generators $g_1,\dots,g_n$ is the group formed by all words written in a unique way as the product of the generators together with their inverses, using the only cancelation rule $g_ig_i^{-1}=e$; $e$ being the empty word.

The universal C*-algebra of $\FF_n$, $C^*(\FF_n)$, is the completion of the group ring of $\FF_n$ with respect to $\norm{x} = \sup_\pi \norm{x}_{\BB(\HH)}$, where the supremum is taken over all unitary representation of the free group into some $\BB(\HH)$. It has the universal property that any *-homomorphism from the group ring into some $\BB(\HH)$ extends to a *-representation of $C^*(\FF_n)$.

Given two C*-algebras $\AA$, $\CC$, there are two canonical ways to turn their algebraic tensor product $\AA \otimes \CC$ into a C*-algebra. Consider two *-representations $\pi_A : \AA \to \BB(\HH_A)$, $\pi_B : \CC \to \BB(\HH_B)$, and define the norm of some tensor $x \in \AA \otimes \CC$ as $\norm{\pi_A \otimes \pi_B (x)}_{\BB(\HH_A \otimes \HH_B)}$. The supremum over all such pairs of representations is called the minimal C*-algebraic tensor norm and will be denoted by $\omin$. It can be proven that it is indeed the smallest possible tensor norm for C*-algebras \cite{Takesaki}. We denote by $\AA \omin \CC$ the completion of $\AA \otimes \CC$ with respect to this norm.

To construct the maximal tensor norm, we do not only consider pairs of *-representations $\pi_A,\pi_B$ into the bounded operators on a pair of different Hilbert spaces, but also all pairs of *-representations into a single $\BB(\HH)$, but with the restriction that the range of $\pi_A$ commutes with the range of $\pi_B$. The resulting C*-structure, obtained after completion of $\AA \otimes \CC$, is called the maximal C*-tensor product and will be denoted by $\AA \omax \CC$.  It is the largest possible tensor norm for C*-algebras \cite{Takesaki}.  Note that if either $\AA$ or $\CC$ is finite dimensional, then the two tensor norms induce the same C*-algebraic structure. This is also the defining property of \emph{nuclear} C*-algebras \cite{Takesaki,Pisier_book}. Thus, if either Alice or Bob has a quantum system which is described by a nuclear C*-algebra, then Tsirelson's problem is trivially true.

Kirchberg's QWEP conjecture, equivalent to Connes' embedding problem, states now that there is also only one possible C*-norm on the tensor product of $C^*(\FF_n)$ with itself, meaning that we have for all $n$
\[C^*(\FF_n) \omin C^*(\FF_n) \,=\, C^*(\FF_n) \omax C^*(\FF_n) \,.\]

\section{Non-signalling operator systems and tensor norms}\label{SecNSG}

In order to show the equivalence between Tsirelson's problem and Connes' embedding problem, we first reformulate the setting of Tsirelson
without referring to some particular Hilbert space. We will start by considering the set of ``marginals'' $p(a|x)$ of Alice's
(or Bob's) possible measurements. By definition, this is the set of $N \times K$ matrices with positive entries such that all columns sum up to one. Following \cite{JP} we will define the space $NSG(N,K)$ to be the complex span of this subset of $\mathbb{C}^{NK}$. To be precise, we define the vector space $NSG(N,K)\subset \mathbb{C}^{NK}$ as the space of all matrices $(\la_{x,a})$ such that there exists $\la\in \mathbb{C}$ such that
\[ \sum_{a} \la_{x,a} \lel \la  \quad \mbox{holds for all $x$.}\]
Next, we have to ensure that the probabilities are represented by outcomes of quantum measurements. In order to understand the duality implicitly involved, we have to recall some facts about operator systems. An operator system $X$ is a subspace $X\subset B(\HH)$ with is closed under taking the adjoint and containing the unit \cite{Paulsen_book}. Then the generalized state space is defined as
\[ S(X;B(\HH))\lel \{u:X\to B(\HH): u \mbox{ ucp} \} \quad, \quad
S(X) \lel \bigcup_{H} S(X;B(\HH))
.\]
Here and in the following ucp means unital completely positive. Note that for finite dimensional $\HH$ we find indeed
\[ S(X;B(\ell_2^n))\subset M_n(X^*) \] as a vector space. Following standard operator space terminology
this means we are identifying the matrix structure of the dual space $X^*$. In our situation, it is useful to reverse this
operation. As a vector space we may identify the dual $NSG(N,K)^*$ of $NSG(N,K)$ with the quotient $\mathbb{C}^{NK}/NSG(N,K)^{\perp}$. The orthogonal space $NSG(N,K)^{\perp}$ is easy to calculate and it is given by
\[ NSG(N,K)^{\perp}\lel \{((\overbrace{\la_1,...,\la_1}^K),(\overbrace{\la_2,\cdots \la_2}^K),\cdots, (\overbrace{\la_N,...,\la_N}^K))\pl:\pl \sum_j \la_j=0\} \pl .\]
Then we define the state space
\begin{equation}\label{nsg}
S_\HH(N,K) \lel \{ (\Phi_x)_{x=1}^N \pl:\pl
\Phi_x:\ell_{\infty}^K\to B(\HH) \mbox{ ucp} \} \pl .
\end{equation}
Each such ucp map $\Phi_i$ defines a positive operator valued measure on the Hilbert space $\HH$ and hence a valid quantum measurement. A sequence $(\Phi_x)_{x=1}^N$ then defines a set of $N$ measurements with $K$ outcomes each. Note that, given such a sequence, we may define the linear map $u:NSG(N,K)^*\to B(\HH)$ given by
\[ u\left( (\la_{x,a}) \right) \lel \sum_{x,a} \la_{x,a} \Phi_x(e_a) \pl.\]
Here $e_a$ is the corresponding unit vector in $\ell_{\infty}^K=\mathbb{C}^K$. It is clear that such a $u$ is well-defined. Indeed, we have that for any $\la=(\la_{x,a})$ in $NSG(N,K)^{\perp}$
\[ u(\la) \lel \sum_{x,a} \la_x \Phi_x( e_a)
\lel \sum_x \la_x \Phi_x(1) \lel \sum_x \la_x \lel 0 \pl .\]
In the following we will use the symbol $NSG(N,K)^*$ for the (more or less) concrete operator system defined by \eqref{nsg}. More
precisely, considering the family $$J=\{(\Phi_x)_{x=1}^N| \Phi_x:\ell_{\infty}^K\rightarrow B(\HH) \text{   is  ucp for every  }x\},$$ the operator system structure defined on $NSG(N,K)^*$ is exactly the one defined by the embedding $$\eta:NSG(N,K)^*\hookrightarrow \oplus_{j\in J} B(\HH_j),$$ such that
\[ \eta\left( (\la_{x,a}) \right)(j) \lel \sum_{x,a} \la_{x,a}
\Phi_x(e_a) \quad ,\quad j \lel (\Phi_{x})_{x=1}^N \in S_\HH(N,K) \pl, \] where $\HH_j$ is the Hilbert space associated with the index $j$.


Note that in the sense of the above definition, any set of measurement devices $\{E^x_a\}$, as introduced in Section \ref{Tsir_prob}, defines indeed an element of the state space $S_\HH(N,K)$ with the identification $\Phi_x(e_a) = E^x_a$.

Up to know, we only formalized the situation for one observer. Coming back to the Tsirelson setting, we now associate one copy of $NSG(N,K)^*$ to both Alice and Bob. A measurement on the combined systems will be then an element of the algebraic tensor product $NSG(N,K)^* \otimes NSG(N,K)^*$. In order to turn this linear space into an operator system, we have to define matrix cones on the algebraic tensor product. There are two obvious choices for doing that, and they exactly reflect the two settings connected to Tsirelson's problem.

\begin{defn}
   We call an element $t \in NSG(N,K)^* \otimes NSG(N,K)^*$ of the algebraic tensor product $\omin$-positive, if
   \begin{align*}
       (\Phi_A \otimes \Phi_B)(t) \;\in\; \BB(\HH_A \otimes \HH_B)
   \end{align*}
   is a positive operator for all pairs $(\Phi_A,\Phi_B)\in S_{\HH_A}(N,K)\times S_{\HH_B}(N,K)$ of completely positive maps defined with respect to different Hilbert spaces. We denote by $NSG(N,K)^* \omin NSG(N,K)^*$ the corresponding operator system and call it the minimal tensor product.
\end{defn}It is not difficult to see that $min$ is the minimal norm in the category of operator systems (see \cite{Paulsen_book}).

\begin{defn}
  We call an element $t \in NSG(N,K)^* \otimes NSG(N,K)^*$ of the algebraic tensor product $\ocmax$-positive, if
  \begin{align*}
      (\Phi_A \cdot \Phi_B)(t) \;\in\; \BB(\HH)
  \end{align*}
  is a positive operator for all pairs $(\Phi_A,\Phi_B)\in S_\HH(N,K)\times S_\HH(N,K)$ of completely positive maps with commuting ranges. We denote the corresponding operator system by $NSG(N,K)^* \ocmax NSG(N,K)^*$ and call it the maximal commuting tensor product.
\end{defn}These definitions (for general operator systems) were introduced and nicely discussed in \cite{Paulsen}. We must note that the notation $\ocmax$ here is different from the one used in \cite{Paulsen} ($\otimes_c$). However, we decided to maintain our notation because it is more suitable in this context.

Using the mathematical objects we have defined in this section, we can state the following equivalence to Tsirelson's problem:

\begin{prop}\label{TsirelsonEquivalence}
   Tsirelson's problem has an affirmative solution if and only if
   \begin{align*}
       (NSG(N,K)^* \omin NSG(N,K)^*)_{sa} \;=\; (NSG(N,K)^* \ocmax NSG(N,K)^*)_{sa}\;
   \end{align*}
   holds isometrically for all $(N,K)$. Here $X_{sa}$ denotes the selfadjoint part of the operator system $X$.
\end{prop}

\begin{proof}
Indeed, note that, for any real coefficients $\{M_{x,y}^{a,b}\}_{x,y;a,b=1}^{N;K}$, the element $M=\sum_{x,y;a,b=1}^{N;K}M_{x,y}^{a,b}e_{x,a}\otimes e_{y,b}\in NSG(N,K)^* \otimes_\alpha NSG(N,K)^*$ is selfadjoint for both $\alpha=\min$ and $\alpha=cmax$. Thus, it is normed by states. The definitions of $NSG(N,K)^*$ and the $min$ and $cmax$ norms yield the assertion.
\end{proof}

Thus, the previous proposition says that Tsirelson's problem is equivalent to check whether two real Banach spaces coincide isometrically.

\section{Connes' embedding problem equals Tsirelson's problem}\label{Sec_equiv}

In this section, we will combine the insights from the previous sections and show that, in some sense, both problems, Tsirelson's and Connes', are indeed equivalent. More specifically, we will prove the following theorem:

\begin{thm}\label{equivalence}
   Let $C^*(\FF_N)$ be the universal C*-algebra of the free group of order $N$, and let $NSG(N,K)^*$ be the non-signaling operator system of order $(N,K)$. Then the following are equivalent.
   \begin{enumerate}
   \item $C^*(\FF_N) \omin C^*(\FF_N) \, = \, C^*(\FF_N) \omax C^*(\FF_N)$ for all $N$,
   \item $NSG(N,K)^* \omin NSG(N,K)^* \,=\, NSG(N,K)^* \ocmax NSG(N,K)^*$ completely isometrically for all $N,K$.
   \end{enumerate}
\end{thm}

Then, the first implication stated in Theorem \ref{FirstTheorem} is immediately obtained from Theorem \ref{equivalence}. Indeed, we have:

\begin{cor}
If Connes' embedding problem is true, then Tsirelson's problem has an affirmative answer.
\end{cor}

\begin{proof}
According to Theorem \ref{equivalence}, a positive solution of Connes' problem implies that $NSG(N,K)^* \omin NSG(N,K)^*$ and $NSG(N,K)^* \ocmax NSG(N,K)^*$ are completely isometric. In particular, they are isometric.
\end{proof}

In order to prove Theorem \ref{equivalence} we need to introduce some notation. We refer to \cite[Chapter 8]{Pisier_book} for a more detailed explanation. Given a (discrete) group $G$, we consider the left regular representation of $G$, $\lambda_G:G\rightarrow B(\ell_2(G))$, defined by $\lambda_G(s)(\delta_t)=\delta_{st}$. Then, we can define the reduced C$^*$-algebra of $G$, $C^*_\lambda(G)$, as the C*-algebra generated by $\lambda_G(G)$. That is $C^*_\lambda(G)=C^*(\lambda_G(G))\subset B(\ell_2(G))$. In this work, we will be interested in the particular case $G=\mathbb{Z}_n$. Since it is an abelian group, we have that $C^*_\lambda(\mathbb{Z}_n)=\ell_\infty(\widehat{\mathbb{Z}}_n)=\ell_\infty(\mathbb{Z}_n)=\ell_\infty^n$, where the identification between $\widehat{\mathbb{Z}}_n$ and $\mathbb{Z}_n$ is via the Fourier Transform. Since abelian groups are amenable, we have $C^*_\lambda(G)=C^*(G)$ (see Section \ref{secqwep} for the definition of the universal C*-algebra of $G$, $C^*(G)$). Given two groups $G_1, G_2$ (resp. $C^*$-algebras $\mathcal A_1, \mathcal A_2$), we will denote by $G_1* G_2$ (resp. $\mathcal A_1*\mathcal A_2$) the free product group (resp. $C^*$-algebra). It is well known that $C^*(G_1*G_2)=C^*(G_1)*C^*(G_2)$. Actually, this identification is true for an arbitrary family of groups $(G_i)_{i\in I}$. Note that we have

$$C^*(\star_{i\in I} \mathbb{Z}_n)=\star_{i\in I}C^*(\mathbb{Z}_n)= \star_{i\in I}\ell_\infty^ n.$$


We will start by stating Connes' embedding problem in terms of $\star_{i\in I}\ell_\infty^ n$. Although the result is certainly known to experts, we will prove it for the sake of completeness.

\begin{lem}\label{freegroup}
$C^*(\FF_N)\omin  C^*(\FF_N) =  C^*(\FF_N)\omax  C^*(\FF_N)$ for all $N$ iff for all $N,K \in \mathbb{N}$ the minimal and the maximal tensor product coincide on the tensor product of $\star_{x=1}^N\ell_\infty^K$ with itself:
\begin{align}\label{eqNM}
  \star_{x=1}^N\ell_\infty^K \omin \star_{x=1}^N\ell_\infty^K \; = \; \star_{x=1}^N\ell_\infty^K \omax \star_{x=1}^N\ell_\infty^K \;.
\end{align}
\end{lem}

\begin{proof} First, we show that \eqref{eqNM} implies \[ C^*(\FF_N)\omin  C^*(\FF_N) =  C^*(\FF_N)\omax  C^*(\FF_N) \, .\]
   Indeed, we observe that $C^*(\FF_N)=\ast_{N} C(\mathbb{T})$. However, $C(\mathbb{T})$ is a commutative C*-algebra and it admits a net $\Psi_{\lambda}$ of completely positive unital maps converging pointwise to the identity. Each of them admits a factorization $\Psi_{\la}=v_{\la}u_{\la}$. Here $u_{\la}:C(\mathbb{T})\to \ell_{\infty}^{m_\lambda}$ and $v_{\la}:\ell_{\infty}^{m_{\la}}\to C(\mathbb{T})$ are both completely positive unital maps (see \cite{Pisier_book}). Recall that unital completely positive maps extend to free products (see \cite{MR1038498,MR1111181}). Therefore we deduce from our assumption that for every $z\in C^*(\FF_N)\otimes C^*(\FF_N)$ and $\lambda$ we have
    \begin{align*}
        &\|(\ast_N \Psi_{\la}\otimes \ast_N \Psi_{\la})(z)\|_{C^*(\FF_N) \otimes_{\max} C^*(\FF_N)}\\
        &\le \|((\ast_N u_{\la}) \otimes (\ast_N u_{\la}))(z)\|_{\star_{x=1}^N\ell_\infty^{m_\lambda} \otimes_{\max} \star_{x=1}^N\ell_\infty^{m_\lambda}}
        \\&=  \|((\ast_N u_{\la}) \otimes (\ast_N u_{\la}))(z)\|_{\star_{x=1}^N\ell_\infty^{m_\lambda} \otimes_{\min} \star_{x=1}^N\ell_\infty^{m_\lambda}}
        \\
        &\leq \|z\|_{C^*(\FF_N) \otimes_{\min} C^*(\FF_N)}.
    \end{align*}
It is easy to see that for tensors $z=\sum_k a_k\otimes b_k$ we have
norm convergence along the net. By density we obtain the assertion.

For the converse implication, recall that a unital C*-algebra $A$ has the Local Lifting Property (LLP) if for every C*-algebra $B$, any (closed two-side) ideal $I\subset B$, any ucp map $u:A\rightarrow B/I$ and any finite dimensional subspace $E\subset A$, there is a complete contraction $\tilde{u}:E\rightarrow B$ that lifts $u|_E:E\rightarrow B/I$. Note that $\ell_\infty$ has the LLP. Furthermore, according to \cite{Pisier} $\star_{x=1}^N\ell_\infty^K$ has the LLP. Now, according to \cite[Proposition 16.13]{Pisier_book}, the QWEP conjecture implies that $A\otimes_{min} A= A\otimes_{max} A$ for every C*-algebra having the LLP. So we conclude the proof.
\end{proof}

The key point in the proof of Theorem \ref{equivalence} is Kasparov's dilation Theorem. We refer to \cite[Chapter 6]{Lance_book} for all missing details. Given a C*-algebra $B$, let us start with the $C^*$-module
\begin{align*}
H_B=\{(x_n)_n: x_n\in B, \sum_nx_n^*x_n \text{   converges in   } B\}.
\end{align*}Note that $H_B$ is a $B$ right module with respect to $(x_n)b=(x_nb)_n$.

A right module map $T:H_B\rightarrow H_B$ is called \emph{adjointable} if there exists a linear map $S:H_B\rightarrow H_B$ such that
\begin{align*}
\langle \xi_1, T(\xi_2)\rangle=\langle S(\xi_1), \xi_2\rangle \text{  for every  } \xi_1,\xi_2\in H_B,
\end{align*}where $\langle (x_n)_n,(y_n)_n\rangle= \sum_nx_n^*y_n$. Then,
\begin{align*}
\mathcal L(H_B):=\{T:H_B\rightarrow H_B| T \text{  adjointable right module map  }\}
\end{align*} is again a C*-algebra. We refer to \cite{Lance_book} for $\mathcal L(H_B)=M(K\otimes B)$, where $M(A)$ is the \emph{multiplier algebra} of a (non- unital) C*-algebra $A$ and $K=K(\ell_2)$ denotes the space of compact operators on $\ell_2$. If we assume $A\subseteq B(H)$ for some Hilbert space, then $T\in B(H)$ belongs to $M(A)$  if $Ta\in A$ and $aT\in A$ for all $a\in A$.

In our particular situation, the C*-algebra $B$ will be unital. Then, it is easy to see that every $T\in M(K\otimes B)$, is represented by a matrix $T=[T_{i,j}]$ with coefficients $T_{i,j}\in B$. It follows, in particular, that
\begin{align}\label{bicommutant}
M(K\otimes B)\subseteq B(\ell_2)\overline{\otimes} B'',
\end{align}where $\overline{\otimes}$ denotes the von Neumann tensor product and $B''$ the bicommutant of the C*-algebra $B$.

\begin{thm}\cite[Kasparov's dilation Theorem]{Lance_book}\label{Kasparov}
Let $A$ and $B$ be separable unital $C^*$- algebras and let $\rho:A\rightarrow B$ ucp. Then, there exists a $*$-homomorphism $\pi:A\rightarrow M(K\otimes B)$ such that $\rho(a)=\pi(a)_{11}=e_{11}\pi(a)e_{11}$ for every $a\in A$.
\end{thm}

Note that this statement is not exactly \cite[Theorem 6.5]{Lance_book}. However, it can be easily obtained from it. Indeed, for any $\rho:A\rightarrow B$ ucp, we can consider $\sigma\circ \rho:A\rightarrow \mathcal L(H_B)$, where $\sigma: B\hookrightarrow \mathcal L(H_B)$ is the canonical embedding defined by $\sigma(b)=1\otimes b$. Then, according to \cite[Theorem 6.5]{Lance_book} we obtain an $*$-homomorphism $\pi:A\rightarrow M_2(\mathcal L(H_B))$ such that $1\otimes \rho(a)=\pi(a)_{11}=e_{1}\pi(a) e_{1}$. By the explanation above, we can see $\rho(a)=(e_{1}\otimes f_{1})\pi(a)(e_{1}\otimes f_{1})$ where here $(f_{n})_n$ denotes the canonical basis of $\ell_2$. Furthermore, as it is explained in \cite[pag. 65]{Lance_book}, there is a canonical identification $M_2(\mathcal L(H_B))\simeq M(K\otimes B)$. Therefore, we can see $\pi:A\rightarrow M(K\otimes B)$ and write $(\pi(a))_{11}:=e_{11}\pi(a)e_{11}$, where we denote $e_{11}\otimes e_{11}$ the rank one projection obtained from $(e_{1}\otimes f_{1})\otimes (e_{1}\otimes f_{1})$ by the previous identification.

\

With this at hand, we can prove the following proposition, which is a crucial point in this work.

\begin{prop}\label{embedNSGfree}
   The space $NSG(N,K)^*$ embeds completely isometrically into the $N$-fold free product of $\ell_\infty^K$ with itself via the map
   \begin{align}\label{eqembedNSG1}
       \iota: NSG(N,K)^* \hookrightarrow \star_{x=1}^N\ell_\infty^K,
   \end{align}defined by $\iota(e_{x,a})=\pi_x(e_a)$ for every $x,a$. Here, $\pi_i:\ell_\infty^K\hookrightarrow *_{x=1}^N\ell_\infty^K$ denotes the natural embedding in the position $x=i$.

   Furthermore, the minimal (resp. the maximal commuting) tensor product of $NSG(N,K)^*$ with itself embeds completely isometrically into the minimal (resp. the maximal) C*-tensor product of $*_{x=1}^N\ell_\infty^K$ with itself via the map
   \begin{align}\label{eqembedNSG2}
       \iota\otimes \iota: NSG(N,K)^* \omin NSG(N,K)^* \hookrightarrow \star_{x=1}^N\ell_\infty^K \omin \star_{x=1}^N\ell_\infty^K, \, \\
       \iota\otimes \iota: NSG(N,K)^* \ocmax NSG(N,K)^* \hookrightarrow \star_{x=1}^N\ell_\infty^K \omax \star_{x=1}^N\ell_\infty^K\;.
   \end{align}
\end{prop}

With this at hand, and using Lemma \ref{freegroup}, we immediately deduce the implication $a) \Rightarrow b)$ in Theorem \ref{equivalence}.

\begin{proof}
   It is very easy to see that the map $\iota: NSG(N,K)^* \hookrightarrow \star_{x=1}^N\ell_\infty^K$ is well defined. Actually, by the very definition of $NSG(N,K)^*$ it follows that $\iota$ is completely positive and unital. Thus, it is a completely contraction. Therefore, in order to prove \eqref{eqembedNSG1} it suffices to show that each completely positive unital map $\Phi$ from $NSG(N,K)^*$ to $\BB(\HH)$ extends to a completely positive map from $\star_{x=1}^N\ell_\infty^K$ into $\BB(\HH)$. By the definition of the operator system $NSG(N,K)^*$, each $\Phi$ is given by a set of $N$ completely positive maps $\Phi_x:  \ell^K_\infty \to \BB(\HH)$. Then, if we consider the separable C*-algebra $A\subset\BB(\HH)$ generated by the $\Phi_x(e_a)$'s, we can apply Theorem \ref{Kasparov} for every $x$ to get a set of unital $*$-representations
\begin{align*}
\pi_x:  \ell^K_\infty \to M(K\otimes A)
\end{align*}such that $\Phi_x(e_a) = (\pi_x(e_a))_{11}=e_{11}\pi_x(e_a)e_{11}$.

   Let $\pi$ be the $*$-representation of $\star_{x=1}^N\ell_\infty^K$ which restricts to $\pi_x$ if we only consider the $x$-th copy of $\ell^K_\infty$. Then the fact that the projection $e_{11}\otimes e_{11}$ is independent of $x$, guarantees that the map
   \[\tilde{\Phi}(\cdot) \, = \, e_{11} \pi(\cdot) e_{11} : \star_{x=1}^N\ell_\infty^K \to \BB(\HH)\]
   is a completely positive unital extension of $\Phi$.

To prove the second part of the proposition, note that the first embedding follows from the injectivity of the minimal tensor product, see for example \cite{Paulsen}. On the other hand, the second inclusion requires a more careful treatment. Consider a pair $\Upsilon=((\Phi_x)_x,(\Psi_y)_y)$ in the state space $S_H(N,K)$ such that
   \[ [\Phi_x(e_a),\Psi_y(e_b)] \lel 0 \] holds for all $x,a$, $y,b$.

Let $A$ be the separable $C^*$-algebra generated by the $\Phi_x(e_a)$'s and $B$ be the separable C*-algebra generated by the $\Psi_y(e_b)$'s. Since all these elements are selfadjoint we deduce that still $[a,b]=0$ holds for elements $a\in A'':=N_A\subseteq \BB(\HH)$, $b\in B'':=N_B \subseteq \BB(\HH)$.

For fixed $x$ we apply again Theorem \ref{Kasparov} and find a representation
   \[ \pi_x: \ell_{\infty}^K \to M(K\otimes A) \]
   such that $\Phi_x(e_a)\lel e_{11}\pi_x(e_a)e_{11}$ for every $a$. According to (\ref{bicommutant}), $$M(K\otimes  A) \subset B(\ell_2)\bar{\otimes} N_A$$
   holds for the von Neumann algebra tensor product. Now call
   \[ \tilde{\pi}_x: \ell_{\infty}^K\to B(\ell_2)\bar{\otimes}B(\ell_2)\bar{\otimes} N_A\pl ,\pl \tilde{\pi}_x(e_a) \lel 1 \otimes  \pi_x(e_a). \] We proceed analogously for every $y$ and define
   \[ \tilde{\sigma}_y: \ell_{\infty}^K\to B(\ell_2)\bar{\otimes}B(\ell_2)\bar{\otimes} N_B\pl ,\pl  \tilde{\sigma}_y(e_{b}) \lel \flip \circ (1 \otimes \pi_y(e_b))\circ \flip.  \pl \]
   Here $\flip(T\otimes  S)=S\otimes  T$ makes sure that the extra $B(\ell_2)$ part is put in the second copy.

Since elements in $N_A$, $N_B\subset \BB(\HH)$ commute, we have $[\tilde{\pi}_x(e_{a}),\tilde{\sigma}_y(e_{b})]\lel 0$ for all $x,a$, $y,b$. Therefore, we can obtain representations $$\pi_1, \pi_2:\star_{x=1}^N\ell_\infty^K\rightarrow B(\ell_2)\bar{\otimes}B(\ell_2)\bar{\otimes} \BB(\HH)$$ with commuting range. This allows us to obtain a representation
\begin{align*}
\pi_1\otimes \pi_2: \star_{x=1}^N\ell_\infty^K\otimes_{max}\star_{x=1}^N\ell_\infty^K \rightarrow B(\ell_2)\bar{\otimes}B(\ell_2)\bar{\otimes} \BB(\HH)
\end{align*}verifying $\pi_1\otimes \pi_2(x\otimes y)=\pi_1(x)\pi_2(y)$ for every $x,y\in \star_{x=1}^N\ell_\infty^K$.

Finally, by defining \[ \tilde{\Upsilon}(\cdot) \lel (e_{11}\otimes e_{11})((\pi_1\otimes \pi_2)(\cdot))(e_{11}\otimes e_{11}) \pl \]we obtain a completely positive unital map on $\star_{x=1}^N\ell_\infty^K\otimes_{max}\star_{x=1}^N\ell_\infty^K$ which extends the initial state $\Upsilon$ on $NSG(N,K)^* \omax NSG(N,K)^*$. This concludes the proof.

\end{proof}

\begin{remark} [projective measurements and POVMs]
\label{vonNeumann}
Since projective measurements correspond to families of representations $\pi_x:\ell_\infty^K\rightarrow B(H)$, we get as a consequence of Proposition \ref{embedNSGfree} that the sets of tensor and commutative quantum correlations $Q$ and $Q'$, defined in Section \ref{Tsir_prob}, are the same whether one considers only projective measurements or also POVMs.
\end{remark}

\begin{remark}
It is also interesting that we have proved a stronger result than
Proposition \ref{embedNSGfree}. Indeed, we have shown that all maps
in the proposition define isometric embeddings in the category of
operator systems. In particular, this means that the natural order
on the space $NSG(N,K)^*$ (and on the corresponding tensor product)
coincides with the natural order on the corresponding C*-algebras.
In \cite{JP} the authors showed that this order is very important in
the context of violation of Bell inequalities.
\end{remark}

We will conclude the paper by showing the implication $b)\Rightarrow
a)$ in Theorem \ref{equivalence}. Note that if the implication
$a)\Rightarrow b)$ is read as: \emph{Connes' embedding problem
implies Tsirelson's problem}, the converse implication can be
understood as: \emph{A positive answer of Tsirelson's problem when
we consider matrix coefficients implies that the QWEP conjecture is
true}. Implication $b)\Rightarrow a)$ follows from the next lemma
joint with Lemma \ref{freegroup}. The proof is based on a trick of
Pisier which can be found in \cite{Pisier}.

\begin{lem}\label{tsma} The identity map
\[ id: NSG(N,K)^* \otimes_{max} NSG(N,K)^* \to NSG(N,K)^* \otimes_{min}
NSG(N,K)^* \]
is completely isometric if and only if
\[ \ast_{i=1}^N \ell_{\infty}^K \otimes_{max} \ast_{i=1}^N
\ell_{\infty}^K
\lel \ast_{i=1}^N \ell_{\infty}^K \otimes_{min} \ast_{i=1}^N \ell_{\infty}^K \, .
\]
\end{lem}

\begin{proof}
   Let $\{U_{ik}\}$ be a spanning set of unitary operators in $NSG(N,K)^*$. Then $\{U_{ik}\}$ clearly generates the C*-algebra $\star_{x=1}^N\ell_\infty^K$. We will identify the elements $U^A_{ik} \in NSG(N,K)^* \omin NSG(N,K)^*$ with the elements $U_{ik} \otimes e$ and $U^B_{ik} \in NSG(N,K)^* \omin NSG(N,K)^*$ with the elements $e \otimes U_{ik}$. Then the set $\{U^{A,B}_{ik}\}$ generates the C*-algebra $\star_{x=1}^N\ell_\infty^K \omin \star_{x=1}^N\ell_\infty^K$. Combining the assumption with the conclusion of Proposition \ref{embedNSGfree}, we get a completely positive unital map
  \begin{align*}
  T: NSG(N,K)^* \omin NSG(N,K)^* \to \star_{x=1}^N\ell_\infty^K \omax \star_{x=1}^N\ell_\infty^K.
  \end{align*}
  As a consequence of Proposition \ref{embedNSGfree}, we can assume that $T(U^{A,B}_{ik})$ is a unitary operator in $\star_{x=1}^N\ell_\infty^K \omax \star_{x=1}^N\ell_\infty^K$. Since $T$ is a completely positive and unital map, it is necessarily completely bounded. Thus, it extends to a (completely positive unital) *-representation $\hat{T}$ (see \cite{Pisier}),
  \begin{align*}
      &\hat{T} : \star_{x=1}^N\ell_\infty^K\omin \star_{x=1}^N\ell_\infty^K \to \star_{x=1}^N\ell_\infty^K \omax \star_{x=1}^N\ell_\infty^K \qedhere.
  \end{align*}
\end{proof}

Let $n\in \mathbb{N}$. Then $M_n(NSG(N,K)^* \omin NSG(N,K)^*)$ is
also an operator system, and hence it suffices to check the complete
isometry for selfadjoint elements (or equivalently it suffices to
check for positivity). This leads naturally to the matrix-valued
version of Tsirelson's problem. Indeed, a state for $M_n(NSG(N,K)^*
\omax NSG(N,K)^*)$ is given by a commuting representation
\[ [E_x^a,F_y^b]= 0 \]
on a  Hilbert space $\H$ and  a family $(\xi_i)_{i=1}^n$ with
$\sum_{i=1}^n \|\xi_i\|_H^2=1$. Then we find the matrix of
coefficients
\[ \omega_{i,j,x,y,a,b} = (\xi_i, E_x^aF_y^b\xi_j) \, .\]
Therefore the first condition in Lemma \ref{tsma} asks whether every
such matrix-valued probability $\omega$ lies in the set
\begin{align*}
 S_n^{min}(N,K) = \{ (\xi_i,E_x^a\otimes F_y^b\xi_j)&:
 \xi_i\in H_1\otimes
 H_2\, , \,  \sum_i \|\xi\|_{H_1\otimes H_2}^2=1\, ,\\
&\quad  \sum_x E_x^a=1_{H_1} \,  \sum_y F_y^b=1_{H_2}  \} \, .
\end{align*}

\section{Summary}

We have shown a close connection between Connes' embedding and Tsirelon's problem, relating in this way a major open problem in operator algebras with a basic foundational problem in quantum mechanics. The connection is yet another consequence of the use of operator space techniques in the foundations of quantum mechanics, following the steps already started in \cite{PWPVJ,JPPVW,JPPVW2,JP,ScholzWerner}.

After this work was completed, we learned that one direction (Connes $\Rightarrow$ Tsirelson) was independently obtained by \cite{Fritz}.

\section{Acknowlodgements}

This work was supported in part by Spanish grants I-MATH, MTM2008-01366, S2009/ESP-1594, the European projects QUEVADIS and CORNER, DFG grant We1240/12-1 and National Science Foundation grant DMS-0901457. V.B.S would like to thank Fabian Furrer for stimulating discussions.

\bibliographystyle{amsplain}
\bibliography{bibfile}

\end{document}